\documentclass[aps,prl,groupedaddress,amsmath,amssymb,twocolumn]{revtex4-2}

\usepackage{amsthm}
\usepackage{mathtools}
\usepackage{dsfont}
\usepackage{graphicx}
\usepackage{braket}
\usepackage{booktabs}
\usepackage[ruled,vlined,linesnumbered]{algorithm2e}
\usepackage{optidef}
\usepackage[caption=false]{subfig}
\usepackage{hyperref} 
\hypersetup{
colorlinks=true,
linkcolor=blue,
urlcolor=blue,
citecolor=blue,
}
\usepackage[capitalise]{cleveref}
\usepackage{multirow}
\usepackage{xcolor}
\definecolor{ForestGreen}{rgb}{0.13,0.55,0.13}

\newcommand{\inner}[2]{\langle #1, #2 \rangle}

\def\tr {\mathrm{Tr}}

\allowdisplaybreaks[2]

\begin{document}


\title{
Geometric Optimization over Quantum State Spaces: Tight Uncertainty Relations and Resource Certification
}


\author{Ma-Cheng Yang$^1$ and Cong-Feng Qiao$^{1,2}$}
\email[]{qiaocf@ucas.ac.cn}
\affiliation{$^1$School of Physical Sciences, University of Chinese Academy of Sciences, 1 Yanqihu East Rd, Beijing 101408, China \\ $^2$ICTP-AP, University of Chinese Academy of Sciences, Beijing 100190, China}


\date{\today}

\begin{abstract}
Determining the fundamental limits of nonlinear functionals of quantum measurement statistics is a crucial yet generally intractable nonconvex optimization problem. We introduce a generic support-function-based outer-approximation framework for solving concave-minimization (or convex-maximization) problems over the quantum state space. By mapping the problem onto a reduced $\mathcal{Z}$-space, we characterize the exact quantum boundary through supporting half space derived from the largest eigenvalues of effective observables. This yields an effective method that produces tight bounds for general measurements in finite-dimensional quantum systems with preassigned numerical precision. As an initial application, we recover the exact variance-based uncertainty relations of \href{https://doi.org/10.1103/PhysRevLett.119.170404}{[PRL \textbf{119}, 170404 (2017)]} and efficiently compute optimal entropic uncertainty relations (EURs). Our results reveal that standard analytical and majorization-based EUR bounds are fundamentally loose for generic measurements, and we show that the resulting certified bounds directly enhance quantum steering detection under asymmetric settings. We further apply the framework to determine the maximal athermality resource certifiable from a restricted measurement scenario. Our method thus provides a universal computational tool for exploring the boundaries of quantum state space and the limits of quantum resources.
\end{abstract}


\maketitle


\noindent
\emph{Introduction---}Determining the fundamental limits of physical observables or measurement statistics over all quantum states is a central problem in quantum mechanics and quantum information theory. Formally, a wide array of problems in quantum theory---ranging from formulating state-independent uncertainty relations (URs) and various nonlinear functionals constructed from moments, to exploring limits in quantum thermodynamics---can be cast as a single optimization task: optimizing a real-valued nonlinear functional of expectation values over the entire quantum state space. This amounts to finding the exact bound
\begin{align}
f(\mathcal{A})
=\inf_{\rho}
F\bigl(
\braket{A_1}_{\rho},\cdots,\braket{A_m}_{\rho}
\bigr) \; ,
\label{eq:general_opt_intro}
\end{align}
where $\mathcal{A}=\{A_1,\cdots,A_m\}$ is a finite set of Hermitian operators, $F$ is a real-valued continuous functional and the optimization is performed for all density operators $\rho$ on the underlying Hilbert space.

Physically, the functional $F$ is often related to a resource measure, such as information entropies, variances, or the Kullback-Leibler divergence, which are typically concave or convex, while the feasible set of quantum states is convex. According to convex optimization theory \cite{boyd2004convex}, convex optimization problems---minimizing a convex function over a convex set, where local optimality implies global optimality---can be solved efficiently. By contrast, concave minimization (or equivalently convex maximization) is generally NP hard \cite{horst96}. Its optimum is typically attained on the boundary of the feasible set, whose geometry is complicated and often inaccessible in high-dimensional quantum systems.

State-independent URs provide a paradigmatic example of this difficulty.
As first formalized by Deutsch \cite{deutsch83}, a state-independent UR  
\begin{align}
\mathcal{U}(\mathcal{A},\mathcal{B};\rho) \geq q(\mathcal{A},\mathcal{B}) 
\label{eq:general_ur}
\end{align}
seeks a nontrivial, purely measurement-dependent lower bound $q(\mathcal{A},\mathcal{B}) = \inf_\rho \mathcal{U}(\mathcal{A},\mathcal{B};\rho)$ for an uncertainty measure $\mathcal{U}$, common examples being the sum of entropies or variances. It is worth noting that the theory of majorization lattice also provides a nice tool for formulating state-independent URs \cite{partovi11,friedland13,li19}. While landmark analytical results such as the Maassen-Uffink relations \cite{maassen88} have provided foundational bounds, they are typically tight only for mutually unbiased bases. Significant efforts have since been devoted to tightening the Maassen-Uffink bound \cite{garrett90,sanches-ruiz98,ghirardi03,de08-dhjjc,puchala13,coles14_improved,xiao16,huang21,xie21,rastegin24,huang25}, and to generalizing it to multimeasurement \cite{sanchez93,sanchez-ruiz95,wu09,rudnicki14,liu15,zhang15,xiao16_strong,riccardi17_tight,huang24} and multipartite \cite{renes09,berta10,coles12,ming20,wu22,krawec23,zhang23} scenarios. Nevertheless, for general, asymmetric, or multimeasurement settings, obtaining the exact optimal bound remains an elusive problem \cite{wehner10}, and analytical approximations consistently fail to capture the true quantum limits.

In this Letter, motivated by a variance-specific algorithm introduced by Schwonnek \emph{et al.} \cite{schwonnek17}, we develop a generic support-function-based outer approximation method for solving such nonlinear optimization problems over the quantum state space. By exploiting an affine reduction to a lower-dimensional $\mathcal{Z}$ space and characterizing the feasible region through its dual representation, we obtain an iterative cutting-plane algorithm. The algorithm encloses the true quantum boundary within a sequence of shrinking polytopes, yielding certified lower and upper bounds whose gap quantifies the numerical precision. We apply the method to the sum of variances and recover the result of Ref.~\cite{schwonnek17}. By reformulating tight entropic uncertainty relations (EURs) as the minimization of an entropy functional associated with a single effective positive operator-valued measures (POVM) over the quantum probability space, we show that the proposed method yields tight EUR bounds, substantially improving upon standard analytical and majorization-based bounds. We further demonstrate that these tighter entropic bounds directly enhance steering detection. Finally, we present a quantum thermodynamic application to certifying the maximum accessible athermality resource extractable from a restricted measurement apparatus.

\emph{$\mathcal{Z}$-space representation of the optimization problems---}Any quantum state $\rho$ on a $d$-dimensional Hilbert space $\mathcal{H}$ can be expanded in terms of the identity and the traceless Hermitian generators of $SU(d)$ as
\begin{align}
\rho = \frac{1}{d}\mathds{1} + \frac{1}{2}\mathbf{r}\cdot\boldsymbol{\pi} \; ,
\label{eq:bloch_rep}
\end{align}
where $\boldsymbol{\pi}=(\pi_1,\cdots,\pi_{d^2-1})$ denotes a set of generators normalized as $\tr[\pi_{\mu}\pi_{\nu}]=2\delta_{\mu\nu}$, and $r_{\mu}=\tr[\rho\pi_{\mu}]$. The Bloch representation in \cref{eq:bloch_rep} establishes a one-to-one correspondence between the quantum state space $\mathcal{D}(\mathcal{H})$ and a subset of $\mathbb{R}^{d^2-1}$, known as the Bloch space $\mathcal{B}(\boldsymbol{\pi})$. 

A natural generalization of the Bloch space is obtained by replacing the generator set $\boldsymbol{\pi}$ with an arbitrary set of Hermitian operators $\mathcal{A}=\{A_1,\cdots,A_m\}$ on $\mathcal{H}$. Whereupon, one can define the measurement-induced Bloch space (MIBS) with respect to $\mathcal{A}$ \cite{yang25}
\begin{align}
\mathcal{C}(\mathcal{A}) := \{\boldsymbol{\alpha}_{\mathcal{A}}(\rho)\;|\;\rho\in\mathcal{D}(\mathcal{H})\} \; ,
\end{align}
where $\boldsymbol{\alpha}_{\mathcal{A}}(\rho) = (\braket{A_1}_\rho,\cdots,\braket{A_m}_\rho)^{\top}$ is a real vector in $\mathbb{R}^m$ consisting of the expectation values of the operators in $\mathcal{A}$ for the state $\rho$. The set $\mathcal{C}(\mathcal{A})$ is also known as the joint numerical range (JNR) \cite{szymanski18,szyma20} or quantum convex support \cite{weis11} of Hermitian operators $A_1,\cdots,A_m$, and has been widely used in studies of uncertainty relations, separability, and experimental exploration of high-dimensional quantum state space \cite{schwonnek17,yang25,xie20}.

For convenience, we hereafter drop the subscript $\mathcal{A}$ in $\boldsymbol{\alpha}_{\mathcal{A}}(\rho)$ and $\mathcal{C}(\mathcal{A})$ when there is no confusion. Thus the general optimization problems \cref{eq:general_opt_intro} can be reformulated as the following optimization problems over the MIBS:
\begin{align}
f(\mathcal{A}) = \inf_{\boldsymbol{\alpha}\in\mathcal{C}}F(\boldsymbol{\alpha}) \; .
\label{eq:f_opt_mibs}
\end{align}
The expectation value vector $\boldsymbol{\alpha}(\rho)$ is related to the Bloch vector $\mathbf{r}$ via the affine mapping
\begin{align}
\boldsymbol{\alpha}(\rho) = \boldsymbol{\alpha}(\mathbf{r}) = \mathbf{s} + M\mathbf{r} \; ,
\label{eq:affine_map}
\end{align}
where $\mathbf{s}=(\tr[A_1],\cdots,\tr[A_m])^{\top}/d$ represents the expectation value vector for the maximally mixed state, and $M$ is a real matrix defined by the entries $M_{\mu\nu}=\tr[A_{\mu}\pi_{\nu}]/2$. Consequently, the set $\mathcal{C}(\mathcal{A})$ of accessible expectation value vectors is an affine image of the Bloch space $\mathcal{B}(\boldsymbol{\pi})$. It forms a compact, convex subset of an ellipsoid in $\mathbb{R}^m$ centered at $\mathbf{s}$, with semiaxis lengths $\sigma_i(M)\sqrt{2(d-1)/d}$, where $\sigma_i(M)$ denotes the $i$-th singular value of $M$. Using the singular value decomposition (SVD) of $M$, we can express this relationship as
\begin{align}
\boldsymbol{\alpha}(\mathbf{r}) &= \mathbf{s} + U\Sigma V^\top\mathbf{r} \notag \\
&= \mathbf{s} + Q\mathbf{z}  \notag \\
&= \boldsymbol{\alpha}(\mathbf{z}) \; ,
\end{align}
where $Q$ consists of the first $r=\mathrm{rank}(M)$ columns of $U$, and we define the reduced variable $\mathbf{z}=(\Sigma_rV_r^\top)\mathbf{r}\in\mathbb{R}^r$, with $\Sigma_r$ being the diagonal matrix of nonzero singular values and $V_r^\top$ containing the first $r$ rows of $V^\top$. The optimization is thus recast in terms of $\mathbf{z}$ over the feasible set $\mathcal{Z} = \{ \mathbf{z} \in \mathbb{R}^r \mid \mathbf{s} + Q\mathbf{z}\in \mathcal{C}(\mathcal{A}) \}$. We can therefore reformulate the optimization problem \cref{eq:f_opt_mibs} as
\begin{align}
f(\mathcal{A}) = \inf_{\mathbf{z}\in\mathcal{Z}}F[\boldsymbol{\alpha}(\mathbf{z})] \; .
\end{align}
The feasible set $\mathcal{Z}$ space inherits compactness and convexity from the quantum state space. Notably, this reduction from $\mathcal{C}(\mathcal{A})$ to $\mathcal{Z}$ space significantly reduces the dimension of the optimization problem, which is particularly advantageous in scenarios involving a large number of operators $m$. However, although the dimension is reduced, the boundary of $\mathcal{Z}$ remains implicit. To tackle this, we turn to the dual description of convex sets.

\emph{Dual representation and boundary characterization ---}A fundamental property of closed convex sets is that they are uniquely determined by their dual representation, i.e. the intersection of all closed half-spaces containing them. This geometric insight is rigorously captured by the support function \cite{rockafellar15}. The support function of a convex set $\mathcal{C} \subset \mathbb{R}^n$ in the direction $\mathbf{u} \in \mathbb{R}^n$ is defined as
\begin{align}
\sigma_{\mathcal{C}}(\mathbf{u}) := \sup_{\mathbf{x} \in \mathcal{C}} \inner{\mathbf{u}}{\mathbf{x}} \; .
\end{align}

Crucially, for the set $\mathcal{C}(\mathcal{A})$ of quantum expectation-value vectors, the support function admits an efficient computation. The maximization over MIBS $\mathcal{C}(\mathcal{A})$ translates to maximizing the expectation value of an observable over the set of density matrices:
\begin{align}
\sigma_{\mathcal{C}}(\mathbf{u}) = \lambda_{\max}\left[\Omega(\mathbf{u})\right] \; ,
\end{align}
where $\Omega(\mathbf{u}) \equiv \sum_{i=1}^{m}u_iA_i$ is the effective observable constructed from the measurement operators, and $\lambda_{\max}(\cdot)$ denotes the maximal eigenvalue. 

Since a closed convex set is completely characterized by its support function, we can express $\mathcal{C}$ as the intersection of infinitely many half spaces:
\begin{align}
\mathcal{C} = \left\{\boldsymbol{\alpha}\in\mathbb{R}^m \big| \inner{\mathbf{u}}{\boldsymbol{\alpha}} \leq \sigma_{\mathcal{C}}(\mathbf{u}), \;\forall \mathbf{u}\in \mathbb{R}^m \right\} \; .
\end{align}
Substituting $\boldsymbol{\alpha} = \mathbf{s} + Q\mathbf{z}$ into this description yields the dual representation of the reduced feasible set $\mathcal{Z}$:
\begin{align}
\mathcal{Z} = \left\{\mathbf{z}\in\mathbb{R}^r \big |\inner{\mathbf{u}}{Q\mathbf{z}} \leq \sigma_{\mathcal{C}}(\mathbf{u}) - \inner{\mathbf{u}}{\mathbf{s}}, \;\forall \mathbf{u}\in \mathbb{R}^m \right\} \; .
\label{eq:Z_dual_rep}
\end{align}
Equation (\ref{eq:Z_dual_rep}) reveals that $\mathcal{Z}$ is defined by an intersection of infinitely many linear constraints. This structural insight suggests that we can approximate $\mathcal{Z}$ using a finite subset of these half spaces, thereby enclosing $\mathcal{Z}$ within a sequence of convex polytopes. This strategy forms the cornerstone of the outer-approximation algorithm, which we detail below.

\emph{Outer-approximation algorithm of concave minimization ---}We now consider an important class of nonconvex optimization problems, namely concave minimization, where $F$ is a concave function \footnote{Note that convex maximization is equivalent to concave minimization upon negating the objective function.}. The core idea is to approximate the feasible set $\mathcal{Z}$ by a sequence of shrinking polytopes $\mathcal{Z}_0 \supset \mathcal{Z}_1 \supset \dots \supset \mathcal{Z}$. At each iteration, the algorithm minimizes the objective function over the current outer polytope and adds a cutting plane to refine the approximation.

To initialize the procedure, we construct a primary polytope $\mathcal{Z}_0$. An effective initialization involves box constraints derived from the spectral limits of the measurement operators: $\lambda_{\min} (A_i)\le \alpha_i\le \lambda_{\max}(A_i)$. In terms of the variable $\mathbf{z}$, this corresponds to setting directions $\mathbf{u}=\pm\mathbf{e}_i$ with $\mathbf{e}_i$ being the $i$th standard basis vector, yielding:
\begin{align}
\mathcal{Z}_0 = \left\{\mathbf{z}\in\mathbb{R}^r \;\big|\; \pm(Q\mathbf{z})_i \le \lambda_{\max}(\pm A_i) \mp s_i \right\} \; .
\end{align}
Note that the initial approximation can be further tightened by including pairwise constraints, such as $\alpha_i + \alpha_j \le \lambda_{\max}(A_i + A_j)$.

Lower bound: For any outer polytope $\mathcal{Z}_k$, we seek to minimize the concave objective function $F(\boldsymbol{\alpha}(\mathbf{z}))$,whose minimum over a convex polytope is necessarily attained at one of the vertices according to the theory of convex analysis. This allows us to define a lower bound for the minimal value of $F$ over the feasible set $\mathcal{Z}$:
\begin{align}
f(\mathcal{A}) \geq \min_{\mathbf{z}\in\mathcal{V}_k} F[\boldsymbol{\alpha}(\mathbf{z})] = F[\boldsymbol{\alpha}(\mathbf{z}^*)] =: f_{-}(\mathcal{A}) \; ,
\end{align}
where $\mathcal{V}_k$ denotes the set of vertices of the polytope $\mathcal{Z}_k$, and $\mathbf{z}^*=\operatorname*{arg\,min}_{\mathbf{z} \in \mathcal{V}_k} F(\boldsymbol{\alpha}(\mathbf{z}))$ is the vertex that minimizes the objective function. 

Upper bound: The optimal vertex $\mathbf{z}^*$ of the outer approximation may not physically correspond to a valid quantum state (i.e., $\mathbf{z}^* \notin \mathcal{Z}$). To establish a valid upper bound, we map $\mathbf{z}^*$ to a feasible state inside $\mathcal{Z}$. We define the gradient vector
$\mathbf{g}=\nabla F(\boldsymbol{\alpha})|_{\boldsymbol{\alpha}
=\boldsymbol{\alpha}(\mathbf{z}^*)}$.
For a nonsmooth concave objective, the gradient can be replaced by any supergradient $\mathbf{g}\in\partial^{+}F(\boldsymbol{\alpha}(\mathbf{z}^*))$, provided such a supergradient exists. We then consider the effective Hamiltonian $\Omega(\mathbf{g})=\sum_i g_iA_i$. Let $\ket{\psi_{\min}}$ be a ground state of $\Omega(\mathbf{g})$. The expectation-value vector associated with this physical state is $\boldsymbol{\alpha}_{\text{real}} = \boldsymbol{\alpha}(\ket{\psi_{\min}})$, which provides a valid upper bound:
\begin{align}
f(\mathcal{A}) \leq F(\boldsymbol{\alpha}_{\text{real}}) =: f_{+}(\mathcal{A}) \; .
\end{align}
The gap between these bounds, $\delta = f_{+}(\mathcal{A}) - f_{-}(\mathcal{A})$, quantifies the precision of our estimation.

Refinement via cutting planes: We generate a new linear constraint based on the support function in the direction of steepest descent $-\mathbf{g}$. The new half space is given by $\inner{-\mathbf{g}}{Q\mathbf{z}} \leq \sigma_{\mathcal{C}}(-\mathbf{g}) + \inner{\mathbf{g}}{\mathbf{s}}$. The updated polytope is then defined as
\begin{align}
\mathcal{Z}_{k+1} = \mathcal{Z}_k \cap \left\{\mathbf{z}\in\mathbb{R}^r \big| \inner{-\mathbf{g}}{Q\mathbf{z}} \leq \sigma_{\mathcal{C}}(-\mathbf{g}) + \inner{\mathbf{g}}{\mathbf{s}} \right\} \; .
\end{align}
By construction, this new constraint separates $\mathbf{z}^*$ from $\mathcal{Z}$ while retaining all valid quantum states. The algorithm iterates this procedure by calculating the vertices, updating the certified bounds, and adding supporting cuts until $\delta$ falls below a prescribed tolerance. Detailed pseudocode and a numerical convergence benchmark are provided in the End Matter.

\emph{Applications and discussion---}As a direct example, the variance-based UR is recovered by augmenting the operator family with $A_0=\sum_i A_i^2$; the explicit reduction is given in the End Matter. We now illustrate the broader applicability of our method through EURs, quantum steering, and quantum thermodynamics.

Tight entropic uncertainty relations: Consider $N$ POVMs $\{\mathcal{A}_i\}_{i=1}^N$, where $\mathcal{A}_i$ has $m_i$ outcomes. Their Shannon EUR can be written as
\begin{align}
\sum_{i=1}^{N}H\bigl[\mathbf{p}_{\mathcal{A}_i}(\rho)\bigr]
\geq Nh(\mathcal{E})-N\ln N \; ,
\end{align}
where $h(\mathcal{E})=\inf_{\rho}H[\mathbf{p}_{\mathcal{E}}(\rho)]$. Here, $\mathcal{E}=\{E_\mu\}_{\mu=1}^{m}$, with $m=\sum_i m_i$, is the effective POVM obtained by concatenating the original POVMs and rescaling their elements by $1/N$; see Supplemental Material \cite{SM_Note} for details and extensions to other entropies. Thus, finding the tight EUR reduces to minimizing the entropy of a single POVM.

For $\mathcal{E}$, the MIBS reduces to the quantum probability space
\begin{align}
\mathcal{P}(\mathcal{E})
:=
\{\mathbf{p}_{\mathcal{E}}(\rho)\mid
\rho\in\mathcal{D}(\mathcal{H})\}
\subseteq\Delta^m \; .
\end{align}
Here, $\Delta^m$ is the $m$-dimensional probability simplex. The optimization can therefore be expressed as
\begin{align}
h(\mathcal{E})
=
\inf_{\mathbf{z}\in\mathcal{Z}}H[\mathbf{p}(\mathbf{z})] \; ,
\end{align}
where the reduced $\mathcal{Z}$-space is $\mathcal{Z} = \{ \mathbf{z} \in \mathbb{R}^r \mid \mathbf{s} + Q\mathbf{z}\in \mathcal{P}(\mathcal{E}) \}$. Since the Shannon entropy is concave, this is a concave-minimization problem of the form considered above.

\begin{figure*}[ht]
\subfloat[Two-measurement setting $\mathcal{M}_2$.]
{\label{fig:qutrit_two_measurement}
\includegraphics[width=0.5\textwidth]{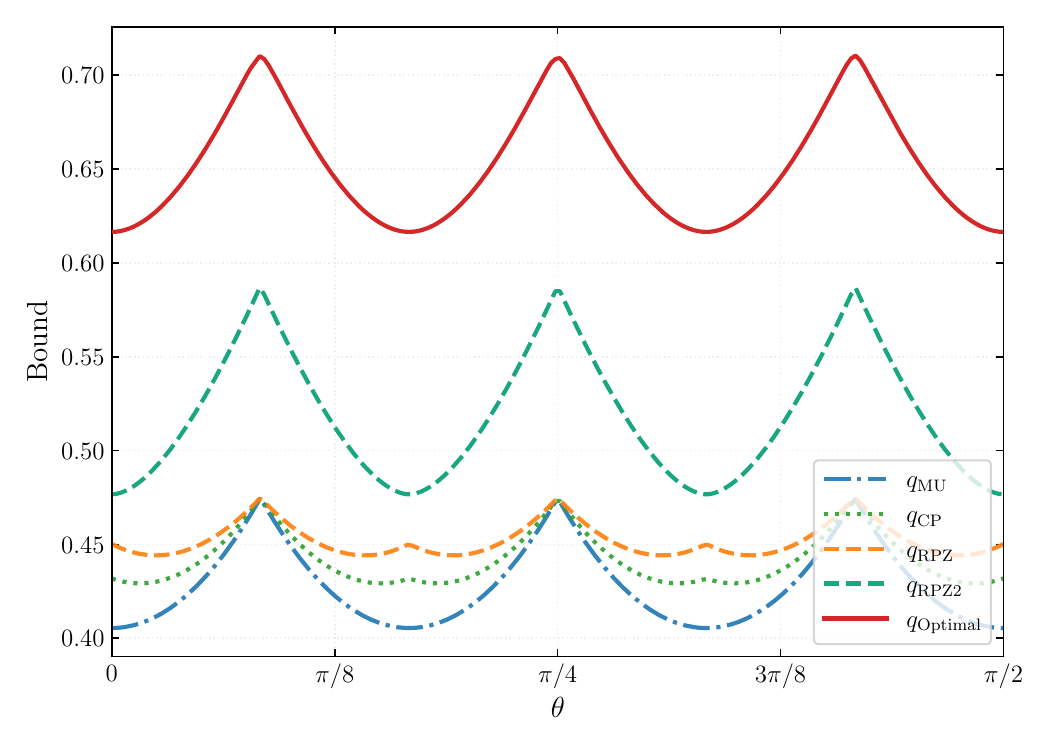}}
\subfloat[Three-measurement setting $\mathcal{M}_3$.]
{\label{fig:qutrit_three_measurement}
\includegraphics[width=0.5\textwidth]{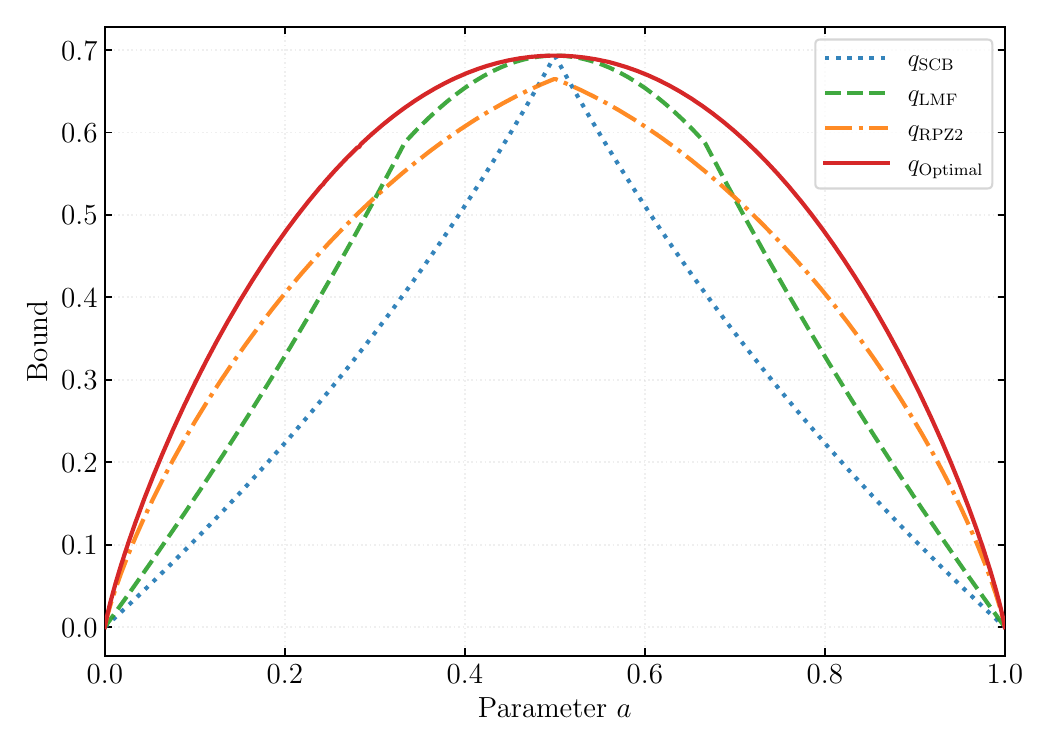}}
\caption{Entropic uncertainty bounds for two- and three-measurement settings. Our method (red) yields the tight quantum-mechanical bound and improves substantially on existing analytical and majorization-based bounds.}
\label{fig:qutrit_steering_general}
\end{figure*} 

We compare our optimal bounds with the Maassen-Uffink (MU), Coles-Piani (CP), and Rudnicki-Puchała-Życzkowski (RPZ) bounds for two measurements \cite{coles14_improved,rudnicki14} and with the RPZ2, SCB, and LMF bounds for multiple measurements \cite{rudnicki14,liu15}. We consider qutrit settings $\mathcal{M}_2$ and $\mathcal{M}_3$, whose bases are specified in the End Matter. As shown in \cref{fig:qutrit_two_measurement,fig:qutrit_three_measurement}, existing bounds are suboptimal over almost the entire parameter range, whereas our method traces the optimal boundary. This shows that maximal overlaps, the second-largest overlap $c_2$, and even majorization constructions involving additional overlaps generally do not fully characterize the boundary of the quantum probability space.

Steering detection with asymmetric settings: Tight EURs also improve the detection of quantum correlations \cite{coles17,uola20}. A nonsteerable state must satisfy \cite{costa18}
\begin{align}
\frac{1}{\alpha-1}
\left[
\sum_{k=1}^{N}
\left(
1-\sum_{i,j}
\frac{(p_{ij}^{(k)})^\alpha}
{(p_i^{(k)})^{\alpha-1}}
\right)
\right]
\geq
q_{\alpha}^{T}(\mathcal{A}_1,\ldots,\mathcal{A}_N) \; , \notag
\end{align}
where the right-hand side is the Tsallis EUR bound. For isotropic states \cite{horodecki99} and $\alpha=2$, this condition reduces to
\begin{align}
\eta
\leq
\sqrt{
1-\frac{d\,q_{2}^{T}(\mathcal{A}_1,\ldots,\mathcal{A}_N)}
{N(d-1)}
} \; .
\end{align}
Hence, exceeding this visibility threshold certifies steerability.

\begin{figure}[t]
\centering
\includegraphics[width=0.5\textwidth]{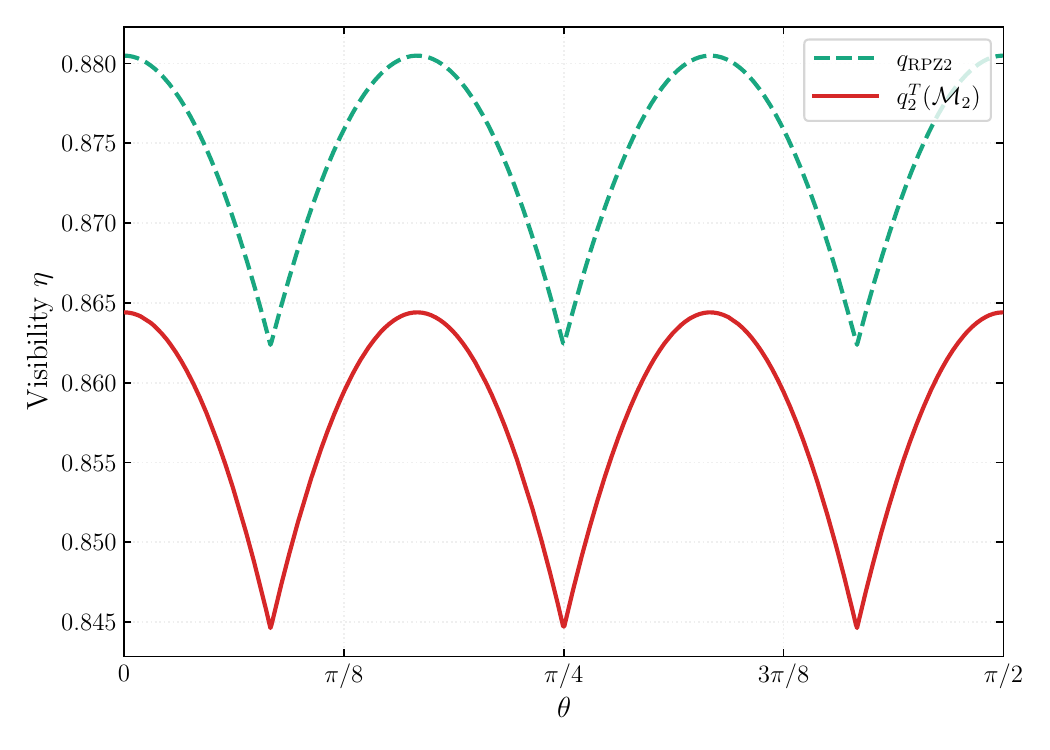}
\caption{Steering thresholds for isotropic states under $\mathcal{M}_2$. Lower thresholds indicate greater robustness to white noise.}
\label{fig:qutrit_steering_two_measurement}
\end{figure}

This improvement is particularly relevant for asymmetric settings, which naturally arise from experimental imperfections. For $\mathcal{M}_2$, \cref{fig:qutrit_steering_two_measurement} shows that our optimal EUR yields a lower noise threshold than the majorization-based bound \cite{rudnicki14}, thereby certifying steerability in a broader parameter regime.

Certifiable athermality under restricted measurements: Consider a finite-dimensional system with Hamiltonian $H$ and Gibbs state $\tau_\beta=e^{-\beta H}/Z_\beta$, where $Z_\beta=\tr[e^{-\beta H}]$ and $\beta=(k_BT)^{-1}$. Its nonequilibrium free-energy difference satisfies
\begin{equation}
F(\rho)-F(\tau_\beta)
=
\beta^{-1}D(\rho\Vert\tau_\beta) \; ,
\end{equation}
and characterizes the extractable work in the thermodynamic limit \cite{donald87,brandao13,horodecki13}.

Suppose that only a fixed POVM $\mathcal{E}=\{E_i\}_{i=1}^m$ is accessible, producing
\begin{equation}
\mathcal{E}(\rho)
=
\mathbf{p}_{\mathcal{E}}(\rho) \; ,
\qquad
[\mathbf{p}_{\mathcal{E}}(\rho)]_i
=
\tr[E_i\rho] \; .
\end{equation}
The athermality certifiable from these statistics is
\begin{equation}
W_{\mathcal{E}}^{\rm obs}(\rho)
=
\beta^{-1}
D_{\rm KL}\bigl(
\mathcal{E}(\rho)\Vert\mathcal{E}(\tau_\beta)
\bigr) \; ,
\end{equation}
where $D_{\rm KL}(\mathbf{p}\Vert\mathbf{q})=\sum_i p_i\log(p_i/q_i)$ is Kullback-Leibler divergence. This is the restricted-measurement analog of the relative entropy of athermality \cite{schindler25}.

The maximum athermality detectable by $\mathcal{E}$ is
\begin{align}
\label{eq:W_obs_max}
W_{\max}^{\rm obs}(\mathcal{E}) =
-\beta^{-1}
\min_{\mathbf{z}\in\mathcal{Z}}
\left[
-D_{\rm KL}\bigl(
\mathbf{p}(\mathbf{z})\Vert\mathcal{E}(\tau_\beta)
\bigr)
\right] \; .
\end{align}
Because $D_{\rm KL}$ is convex in its first argument, this is a concave minimization of $-D_{\rm KL}$, and can be solved by our method.

As an example, we consider a five-qubit mixed-field Ising chain,
\begin{equation}
H_N
=
-J\sum_{j=1}^{N-1}
\sigma_z^{(j)}\sigma_z^{(j+1)}
-g\sum_{j=1}^{N}\sigma_x^{(j)}
-h\sum_{j=1}^{N}\sigma_z^{(j)},
\end{equation}
under a noisy measurement model detailed in Supplemental Material \cite{SM_Note}. Figure~\ref{fig:w_heatmap_beta_eta} shows $W_{\max}^{\rm obs}(\mathcal{E})$ as a function of the inverse temperature $\beta$ and measurement sharpness $\eta$, illustrating their joint effect on the maximum certifiable athermality.

\begin{figure}[t]
\centering
\includegraphics[width=0.5\textwidth]{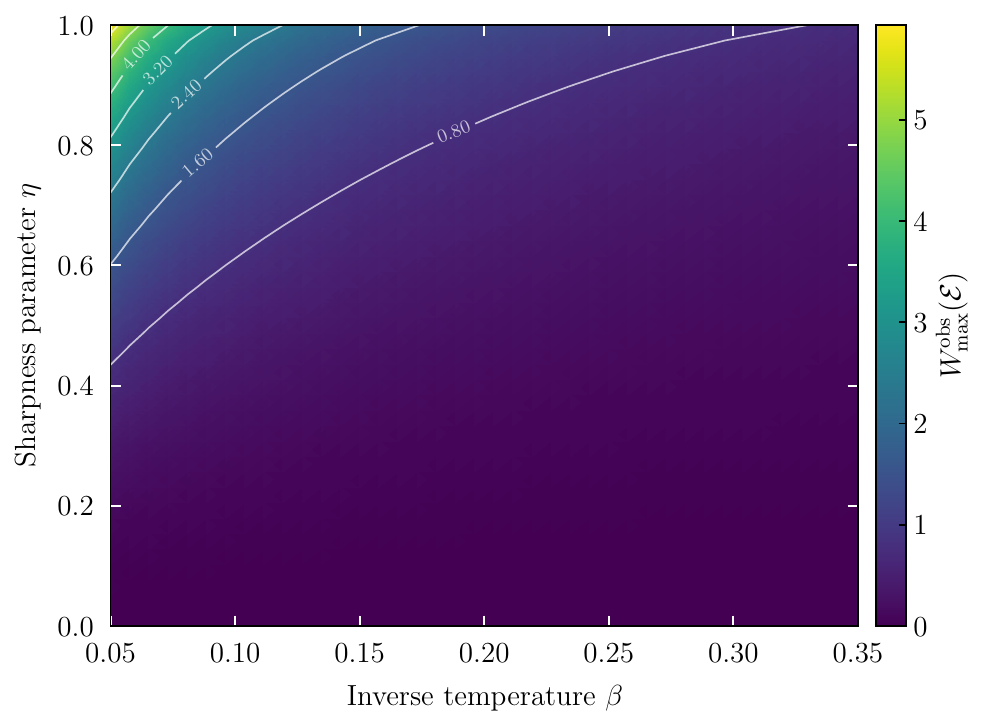}
\caption{Maximum certifiable athermality $W_{\max}^{\rm obs}(\mathcal{E})$ for a five-qubit system with $J=1$, $g=0.7$, and $h=0.2$.}
\label{fig:w_heatmap_beta_eta}
\end{figure}

\emph{Conclusion and outlook---}We have developed a support-function-based outer-approximation method for concave minimization (equivalently, convex maximization) over the quantum state space. By mapping the problem to a reduced $\mathcal{Z}$ space and approximating its feasible region with supporting half spaces, the method converts the original optimization into an iterative cutting-plane procedure. Each cut requires only the largest eigenvalue of an effective observable, while every iteration provides certified lower and upper bounds that converge to the optimum within any prescribed precision.

Our framework recovers known variance-based uncertainty relations \cite{schwonnek17} and provides certified tight bounds for EURs by reducing them to entropy minimization over the quantum probability space of a single effective POVM. The results show that standard analytical and majorization-based bounds can be substantially suboptimal, particularly for generic measurement settings. The resulting tighter bounds improve steering detection under asymmetric measurements. We further demonstrated the method in quantum thermodynamics by computing the maximum certifiable athermality accessible to a restricted measurement apparatus.

Computationally, the method relies on the support-function oracle $\mathbf{u}\mapsto\sigma_{\mathcal{C}}(\mathbf{u})$, evaluated through a largest-eigenvalue computation. Successive cuts generate outer polytopes converging to the feasible set and produce monotonically tightening certified bounds. The geometric complexity is controlled by the reduced affine dimension $r=\operatorname{rank}(M)\leq\min\{m,d^2-1\}$, rather than directly by
the full state-space dimension. The spectral-oracle cost still depends on the Hilbert-space dimension $d$, while vertex enumeration can become a bottleneck when $r$, and hence the number of measurements, is large. Nevertheless, every finite iteration remains rigorous and provides valid bounds. In high-dimensional settings, the method may instead be applied to multiple fixed-size measurement slices, yielding complementary certified projections of the underlying quantum state space.

More broadly, many quantum-information problems amount to optimizing over, or characterizing the boundaries of, compact convex sets induced by quantum-state space. Our approach therefore provides a general geometric and computational framework for studying quantum-state spaces, quantum correlations, uncertainty relations, and thermodynamic resources under restricted measurements.


\emph{Acknowledgments}---This work was supported in part by the National Natural Science Foundation of China (NSFC) under Grants No. 12475087, and No. 12235008, and the China Postdoctoral Science Foundation funded project No. 2024M753174.

\emph{Data availability}---
The data that support the findings of this article are openly available at the repository \cite{SourceCode}.

%

\onecolumngrid
\begin{center}
\textbf{End Matter} 
\end{center}
\twocolumngrid

\emph{Pseudocode and numerical benchmark---}The algorithm constructs a sequence of outer polytopes
$\mathcal{Z}_k\supset\mathcal{Z}$ and evaluates certified lower and upper bounds on the minimum of the concave function $F(\boldsymbol{\alpha})$. At each iteration, a supporting half-space is generated from the ground state of an effective observable. The procedure terminates when the certified gap is below a prescribed tolerance $\epsilon$; see \cref{algor:outer_appr_entropy}.

\begin{algorithm}[!ht]
\caption{\label{algor:outer_appr_entropy}Support-Function-Based Outer Approximation}
\SetAlgoLined
\KwIn{Hermitian operators $\mathcal{A}=\{A_i\}$, concave objective $F$, and tolerance $\epsilon$}
\KwOut{Certified bounds $[f_{-}, f_{+}]$}

\textbf{Initialize:} Construct affine basis $(\mathbf{s}, Q)$\;
\textbf{Initial Polytope:} Build $\mathcal{Z}_0$ with Box/Pair constraints\;
Set $k=0$, Gap $\delta = \infty$\;

\While{$\delta > \epsilon$}{
\tcc{Step 1: Lower Bound via Vertices}
Compute vertices $\mathcal{V}_k$ of $\mathcal{Z}_k$\;
$\mathbf{z}^* \leftarrow \operatorname*{arg\,min}_{\mathbf{z} \in \mathcal{V}_k} F(\mathbf{s} + Q\mathbf{z})$\;
$\boldsymbol{\alpha}_{\text{poly}} \leftarrow \mathbf{s} + Q\mathbf{z}^*$, $f_{-} \leftarrow F(\boldsymbol{\alpha}_{\text{poly}})$\;

\tcc{Step 2: Upper Bound via Spectral Decomposition}
$\mathbf{g} \leftarrow \nabla F(\boldsymbol{\alpha}_{\text{poly}})$\;
Compute $\psi_{\min}$ as ground state of $\Omega(\mathbf{g})=\sum_i g_i A_i$\;
$\boldsymbol{\alpha}_{\text{real}} \leftarrow \boldsymbol{\alpha}(\ket{\psi_{\min}})$\;
$f_{+} \leftarrow F(\boldsymbol{\alpha}_{\text{real}})$\;

\tcc{Step 3: Cut Generation}
$\delta \leftarrow f_+ - f_-$\;
\If{$\delta \le \epsilon$}{
\textbf{break}\;
}
$\sigma \leftarrow -\lambda_{\min}(\Omega(\mathbf{g}))$\;
Add the cut $-(\mathbf g^\top Q)\mathbf z \leq\sigma+\mathbf g^\top\mathbf s$ to obtain $\mathcal Z_{k+1}$\;
$k\leftarrow k+1$\;
}
\Return $[f_-, f_+]$\;
\end{algorithm}

We benchmark the method by minimizing the Shannon entropy of Haar-random POVM with $m=4$ outcomes in dimension $d=100$, generated using \href{https://heyredhat.github.io/qbism/05random.html}{\texttt{qbism.random\_haar\_povm}}. As shown in \cref{fig:vertex_track_random_povm}, the certified gap reaches approximately $10^{-8}$ after $100$ iterations.

\begin{figure}[!t]
\centering
\includegraphics[width=0.5\textwidth]{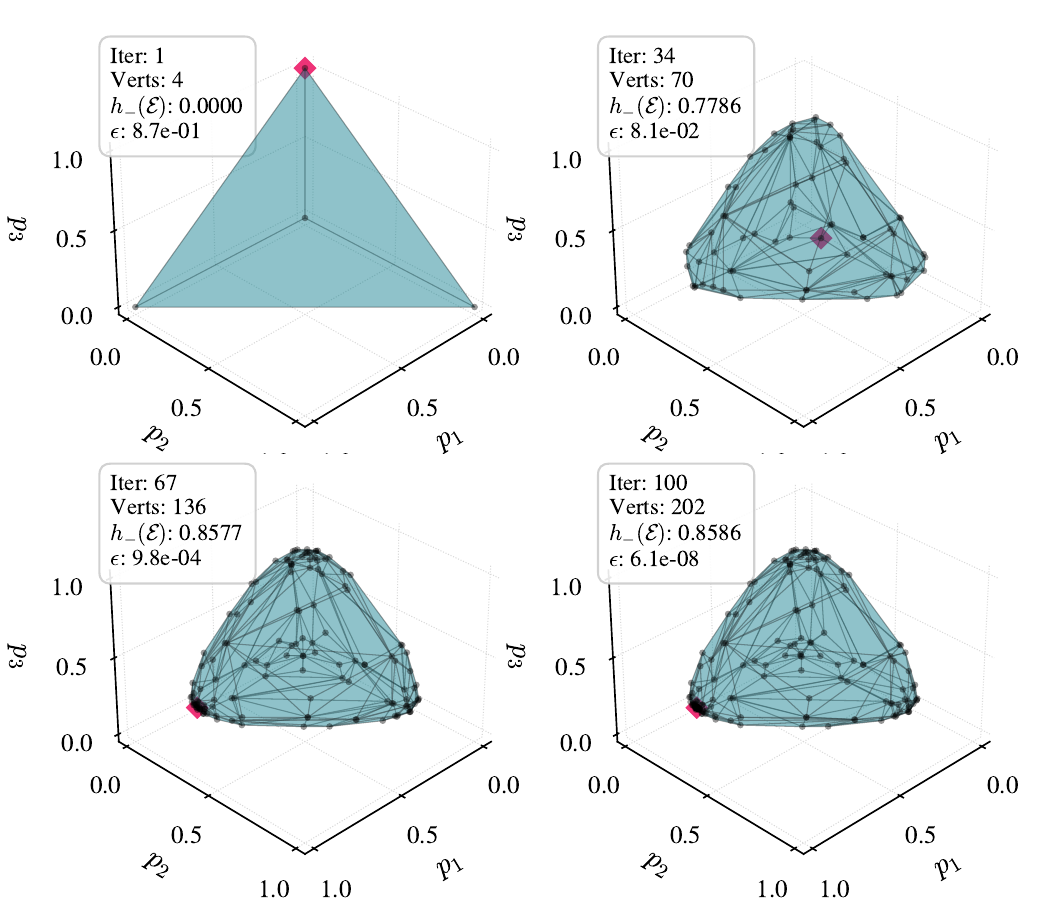}
\caption{Convergence of the outer polytope for Shannon-entropy minimization of a Haar-random POVM with $m=4$ outcomes in dimension $d=100$. The red square marks the minimizing vertex of the current outer polytope.}
\label{fig:vertex_track_random_povm}
\end{figure}

\emph{Variance-based uncertainty relations---}For observables $A_1,\cdots,A_m$, the sum of variances can be written as
\begin{align}
\sum_{i=1}^{m}\Delta_\rho^2(A_i)
&=
\left\langle \sum_{i=1}^{m}A_i^2\right\rangle_\rho
-\sum_{i=1}^{m}\langle A_i\rangle_\rho^2 \notag\\
&=
\alpha_0-\sum_{i=1}^{m}\alpha_i^2
=:F(\boldsymbol{\alpha}) \; ,
\end{align}
where
\begin{equation}
A_0=\sum_{i=1}^{m}A_i^2 \; ,
\qquad
\alpha_i=\langle A_i\rangle_\rho \; .
\end{equation}
Since $F(\boldsymbol{\alpha})$ is concave, the optimal
variance-based uncertainty relation is obtained from
\begin{equation}
f(\mathcal A)
=
\inf_{\mathbf z\in\mathcal Z}
\left[
\alpha_0(\mathbf z)
-\sum_{i=1}^{m}\alpha_i(\mathbf z)^2
\right] \; ,
\end{equation}
where $\mathcal A=\{A_0,A_1,\cdots,A_m\}$. This is a direct instance of our framework and recovers the results of Ref. \cite{schwonnek17}.

\emph{Measurement settings---}The qutrit measurement settings $\mathcal M_2$ and $\mathcal M_3$ used in the main text consist of two and three orthonormal bases, respectively:
\begin{align}
\mathcal M_2:\quad
&\bigl\{(1,0,0),(0,\cos\theta,-\sin\theta) \; , (0,\sin\theta,\cos\theta)\bigr\} \; ,
\nonumber\\
&\frac{1}{\sqrt6}
\bigl\{(\sqrt2,\sqrt3,1),(\sqrt2,0,-2) \; ,
(\sqrt2,-\sqrt3,1)\bigr\} \; ; \nonumber
\label{eq:M2_setting}
\\[1mm]
\mathcal M_3:\quad
&\bigl\{(1,0,0),(0,1,0),(0,0,1)\bigr\} \; ,
\nonumber\\
&\bigl\{(1/\sqrt2,0,-1/\sqrt2),(0,1,0),
(1/\sqrt2,0,1/\sqrt2)\bigr\} \; ,
\nonumber\\
&\bigl\{(\sqrt a,e^{i\phi}\sqrt{1-a},0),
(\sqrt{1-a},-e^{i\phi}\sqrt a,0),
(0,0,1)\bigr\} \; . \nonumber
\end{align}
Here, $0\leq\theta<\pi$, $0\leq\phi<2\pi$, and $0\leq a\leq1$. In the main text, we set $\phi=\pi/2$.

\end{document}



\title{
Geometric Optimization over Quantum State Spaces: Tight Uncertainty Relations and Resource Certification
\\ 
\vspace{2ex}
Supplemental Material
}


\author{Ma-Cheng Yang$^1$ and Cong-Feng Qiao$^{1,2}$}
\email[]{qiaocf@ucas.ac.cn}
\affiliation{$^1$School of Physical Sciences, University of Chinese Academy of Sciences, 1 Yanqihu East Rd, Beijing 101408, China \\ $^2$ICTP-AP, University of Chinese Academy of Sciences, Beijing 100190, China}


\date{\today}


\maketitle

\onecolumngrid
\beginsupplement   


\section{Equivalent form of entropic uncertainty relations}\label{sec:eur}

A POVM $\mathcal{A}=\{E_{\mu}^A\}_{\mu=1}^m$ is defined as a set of positive semi-definite operators that sum to the identity, i.e., $E_{\mu}^A\geq 0$ and $\sum_{\mu=1}^m E_{\mu}^A=\mathds{1}$. This framework provides the most general description of quantum measurements \cite{nielsen10}. According to Born's rule, the probability distribution of outcomes obtained when measuring a state $\rho$ with POVM $\mathcal{A}$ is:
\begin{align}
p(\mu|\mathcal{A}) = \tr[E_{\mu}^A\rho] \; .
\end{align}
Consider a set of $N$ POVMs, denoted $\mathcal{A}_1, \mathcal{A}_2, \cdots, \mathcal{A}_N$, with $m_1, m_2, \cdots, m_N$ outcomes respectively. We define a unified entropic uncertainty functional as follows:
\begin{align}
\mathcal{U}(\mathcal{A}_1,\mathcal{A}_2,\cdots ;\rho) := H\left(\frac{1}{N}\bigoplus_{i=1}^N\mathbf{p}_{\mathcal{A}_i}\right) \; ,
\label{eq:n_povm_entropy_ur}
\end{align}
where $H(\cdot)$ denotes a generalized entropy function (such as Shannon, Tsallis, or Rényi entropy), and $\oplus$ represents the concatenation (direct sum) of probability vectors. It is readily seen that this definition satisfies all the requirements for an uncertainty measure proposed by Deutsch \cite{deutsch83}. 

Crucially, the entropic uncertainty functional in \cref{eq:n_povm_entropy_ur} can be reinterpreted as the entropy of a \textit{single} effective POVM $\mathcal{E}=\{E_{\mu}\}_{\mu=1}^{m}$ with total outcome number $m=\sum_{i=1}^N m_i$. The elements of $\mathcal{E}$ are constructed by scaling and concatenating the original operators:
\begin{equation}
E_{\mu} = 
\begin{cases}
N^{-1} E_{\mu}^{A_1} \; , & 1 \le \mu \le m_1 \; , \\ 
N^{-1} E_{\mu-m_1}^{A_2} \; , & m_1 < \mu \le m_1+m_2 \; , \\ 
\quad \vdots & \quad \vdots \\
N^{-1} E_{\mu-(m-m_N)}^{A_N} \; , & m-m_N < \mu \le m \; .
\end{cases} 
\end{equation}
By construction, $\sum_{\mu=1}^m E_\mu = \frac{1}{N}\sum_{i=1}^N \left(\sum_{\nu} E_\nu^{A_i}\right) = \mathds{1}$, ensuring that $\mathcal{E}$ is a valid POVM. With this mapping, the uncertainty functional simplifies to:
\begin{align}
\mathcal{U}(\mathcal{A}_1,\mathcal{A}_2,\cdots ;\rho) = H(\mathbf{p}_{\mathcal{E}}(\rho)) \; .
\end{align}
Consequently, the problem of deriving the entropic uncertainty relation (EUR) for multiple POVMs is equivalent to minimizing the entropy of this single effective POVM $\mathcal{E}$. We define the minimal entropy as:
\begin{align}
h(\mathcal{E}) := \inf_{\rho}H(\mathbf{p}_{\mathcal{E}}(\rho)) \; .
\end{align}
Based on this reduction, we can formulate EURs corresponding to $\alpha$-order Tsallis entropy:
\begin{align} 
\sum_{i=1}^{N}H_{\alpha}^T\left(\mathbf{p}_{\mathcal{A}_i}(\rho)\right) &\geq q_{\alpha}^{T}(\mathcal{A}_1,\cdots,\mathcal{A}_N) \; , \label{eq:tsallis_eur}
\end{align}
where the bound is given by:
\begin{align}
&q_{\alpha}^{T}(\mathcal{A}_1,\cdots,\mathcal{A}_N) = N^{\alpha}h_{\alpha}^{T}(\mathcal{E}) - \frac{N-N^{\alpha}}{1-\alpha} \; .
\end{align}
Here, $H_{\alpha}^T(\mathbf{p})=\frac{1}{1-\alpha}(\sum_i p_i^\alpha - 1)$ is the Tsallis entropy \cite{tsallis88} with $\alpha\in (0,1)\cup(1,\infty)$ and $h_{\alpha}^{T}(\mathcal{E})=\min_\rho H_{\alpha}^T(\mathbf{p}_{\mathcal{E}}(\rho))$. 
\begin{proof} 
Consider $N$ probability vectors $\mathbf{p}_1, \cdots, \mathbf{p}_N$, and let $\mathbf{q}$ be the effective probability distribution constructed by concatenating and scaling these vectors:
\begin{equation}
\mathbf{q} = \frac{1}{N} \bigoplus_{i=1}^N \mathbf{p}_i \; .
\end{equation}
The components of $\mathbf{q}$ are given by $\{ p_{i,j}/N \}_{i,j}$, where $j$ runs over the outcomes of the $i$-th POVM. Substituting this into the definition of Tsallis entropy, we proceed as follows:
\begin{align}
H_{\alpha}^T(\mathbf{q}) &= \frac{1}{1-\alpha}\left[\sum_{i=1}^N \sum_{j} \left(\frac{p_{i,j}}{N}\right)^\alpha - 1\right] \notag \\ 
&= \frac{1}{1-\alpha}\left[\frac{1}{N^{\alpha}}\sum_{i=1}^N \left(\sum_{j} p_{i,j}^{\alpha}\right) - 1\right] \notag \\ 
&= \frac{N^{-\alpha}}{1-\alpha}\left[\sum_{i=1}^N \sum_{j} p_{i,j}^{\alpha} - N^{\alpha}\right] \; .
\end{align}
To recover the sum of individual entropies, we rewrite the term inside the bracket by adding and subtracting $N$:
\begin{align}
H_{\alpha}^T(\mathbf{q}) &= \frac{N^{-\alpha}}{1-\alpha}\left[\left(\sum_{i=1}^N \sum_{j} p_{i,j}^{\alpha} - N\right) + \left(N - N^{\alpha}\right)\right] \notag \\ 
&= N^{-\alpha} \sum_{i=1}^N \underbrace{\frac{\sum_{j} p_{i,j}^{\alpha} - 1}{1-\alpha}}_{H_{\alpha}^T(\mathbf{p}_i)} + \frac{N^{-\alpha}(N - N^{\alpha})}{1-\alpha} \notag \\ 
&= N^{-\alpha} \sum_{i=1}^N H_{\alpha}^T(\mathbf{p}_i) + \frac{N^{1-\alpha} - 1}{1-\alpha} \; .
\end{align}
Finally, multiplying both sides by $N^{\alpha}$ and rearranging the terms yields the relation used in the derivation of the EUR:
\begin{align}
\sum_{i=1}^N H_{\alpha}^T(\mathbf{p}_i) = N^{\alpha} H_{\alpha}^T(\mathbf{q}) - \frac{N - N^{\alpha}}{1-\alpha} \; .
\end{align}
This confirms that minimizing $H_{\alpha}^T(\mathbf{q})$ is equivalent to minimizing the sum of the individual entropies, up to constant shift and scaling factors. Let $h_{\alpha}^{T}(\mathcal{E})=\min_\rho H_{\alpha}^T(\mathbf{p}_{\mathcal{E}}(\rho))$, we have Tsallis EURs
\begin{align}
\sum_{i=1}^N H_{\alpha}^T(\mathbf{p}_i) \geq N^{\alpha} h_{\alpha}^{T}(\mathcal{E}) - \frac{N - N^{\alpha}}{1-\alpha} \; .
\end{align}
\end{proof}

The Rényi entropy $H^R_{\alpha}(\mathbf{p})=\frac{1}{1-\alpha}\ln\sum_i p_i^\alpha$ is slightly different, since it is Schur concave but not concave over the entire parameter range $\alpha\in (0,1)\cup(1,\infty)$. Specifically, it is concave only for $\alpha\in (0,1)$, but not for $\alpha>1$. Nevertheless, we note that it is related to the Tsallis entropy through
\begin{align}
H_{\alpha}^R = \frac{1}{1-\alpha}\ln \left[1+(1-\alpha)H_{\alpha}^T\right] \; .
\end{align}
Since $\frac{d H_{\alpha}^R}{d H_{\alpha}^T}=\frac{1}{P_\alpha}>0$ with $P_\alpha=\sum_{i}p_i^\alpha$, the Rényi entropy is a strictly monotonically increasing function of the Tsallis entropy. Consequently, the minimal Rényi entropy is completely determined by the minimal Tsallis entropy,
\begin{align}
h_{\alpha}^{R}(\mathcal{E}) = \min_\rho H_{\alpha}^R(\mathbf{p}_{\mathcal{E}}(\rho)) = \frac{1}{1-\alpha}\ln \left[1+(1-\alpha)h_{\alpha}^{T}(\mathcal{E})\right] \; ,
\end{align}
which yields the EUR for the $\alpha$-order Rényi entropy:
\begin{align} 
H_{\alpha}^R\left(\frac{1}{N}\bigoplus_{i=1}^{N}\mathbf{p}_{\mathcal{A}_i}(\rho)\right) &\geq q_{\alpha}^{R}(\mathcal{A}_1,\cdots,\mathcal{A}_N) \; , \label{eq:renyi_eur}  
\end{align}
with
\begin{align}
q_{\alpha}^{R}(\mathcal{A}_1,\cdots,\mathcal{A}_N) = \frac{1}{1-\alpha}\ln \left[1+(1-\alpha)h_{\alpha}^{T}(\mathcal{E})\right] \; .
\end{align}
It is worth noting that, for the Rényi entropy, the term $H^R_{\alpha}((\mathbf{p}\oplus\mathbf{q})/2)$ generally does not decompose into a sum of the individual entropies $H^R_{\alpha}(\mathbf{p})$ and $H^R_{\alpha}(\mathbf{q})$. Therefore, \cref{eq:renyi_eur} defines a distinct type of entropic uncertainty relation. In the limit $\alpha\to 1$, both families recover the Shannon EUR:
\begin{align}
\sum_{i=1}^{N}H\left(\mathbf{p}_{\mathcal{A}_i}(\rho)\right) &\geq Nh(\mathcal{E}) - N\ln N \; .
\end{align}

\section{Certifiable athermality under restricted measurements}

We consider an $N$-spin mixed-field Ising chain with Hamiltonian 
\begin{equation}
H_N = -J\sum_{j=1}^{N-1}\sigma_z^{(j)}\sigma_z^{(j+1)}
-g\sum_{j=1}^{N}\sigma_x^{(j)}
-h\sum_{j=1}^{N}\sigma_z^{(j)} \; .
\end{equation}
To construct a coarse-grained measurement model for detecting athermality, we consider the following three observables
\begin{align}
G_1 &= \frac{1}{N}\sum_{j=1}^{N}\sigma_j^z \; , \\
G_2 &= \frac{1}{N}\sum_{j=1}^{N}\sigma_j^x \; , \\
G_3 &= \frac{1}{N-1}\sum_{j=1}^{N-1}\sigma_j^z\sigma_{j+1}^z \; ,
\end{align}
where $G_3$ captures the nearest-neighbour interaction term of the Hamiltonian. The three observables span a three-dimensional vector space, in which we place four measurement directions at the vertices of a regular tetrahedron,
\begin{equation}
\mathbf{u}_0 = \tfrac{1}{\sqrt{3}}(1,1,1) \; ,\quad
\mathbf{u}_1 = \tfrac{1}{\sqrt{3}}(1,-1,-1) \; ,\quad
\mathbf{u}_2 = \tfrac{1}{\sqrt{3}}(-1,1,-1) \; ,\quad
\mathbf{u}_3 = \tfrac{1}{\sqrt{3}}(-1,-1,1) \; ,
\end{equation} 
which are pairwise separated by the maximal angle
$\arccos(-\tfrac{1}{3})$ and satisfy
$\sum_{k=0}^{3}\mathbf{u}_k = \mathbf{0}$. For each direction we form the Hermitian combination
\begin{equation}
C_k = \sum_{m=1}^{3} u_{k,m}\,G_m \; , \; k = 0,1,2,3 \; .
\end{equation}
The four-outcome POVM is then defined by
\begin{equation}
E_k = \frac{1}{4}\left[\mathds{1} + \eta\, c_{\max}\, C_k\right] \; , \;  k = 0,1,2,3 \; ,
\end{equation}
where $\eta \in [0,1]$ is a sharpness parameter and $c_{\max}$ is the largest scaling factor compatible with positivity. Since each element must satisfy $E_k \geq 0$, i.e.
$1 + c\,\lambda_{\min}(C_k) \ge 0$, the saturating factor is
\begin{equation}
c_{\max} = \min_{k:\,\lambda_{\min}(C_k)<0}
\left( -\frac{1}{\lambda_{\min}(C_k)} \right) \; .
\end{equation}
Using $\sum_{k=0}^{3}\mathbf{u}_k = \mathbf{0}$, hence
$\sum_k C_k = 0$, the completeness relation holds exactly $\sum_{k=0}^{3} E_k
= \frac{1}{4}\left[4\,\mathds{1}
+ \eta\, c_{\max}\sum_{k=0}^{3} C_k\right] = \mathds{1}$ and positivity $E_k\geq0$ is guaranteed for all $\eta \in [0,1]$ by the definition of $c_{\max}$.

\section{Optimization equivalence between MIBS and QPS}

In the main text, we consider general optimization problems of the form
\begin{align}
f(\mathcal{A})
=\inf_{\rho}
F\bigl(
\braket{A_1}_{\rho},\cdots,\braket{A_m}_{\rho}
\bigr) \; ,
\label{sm_eq:general_opt_intro}
\end{align}
where $\mathcal{A}=\{A_1,\cdots,A_m\}$ is a finite set of Hermitian operators, $F$ is a real-valued continuous functional, and the optimization is taken over all density operators $\rho$ on the underlying Hilbert space. This problem can be equivalently reformulated in terms of the measurement-induced Bloch space (MIBS) associated with $\mathcal{A}$,
\begin{align}
\mathcal{C}(\mathcal{A})
:=
\left\{
\boldsymbol{\alpha}_{\mathcal{A}}(\rho)
\;\middle|\;
\rho\in\mathcal{D}(\mathcal{H})
\right\} \subseteq \mathbb{R}^m ,
\end{align}
where
\begin{align}
\boldsymbol{\alpha}_{\mathcal{A}}(\rho)
=
\bigl(
\braket{A_1}_{\rho},
\cdots,
\braket{A_m}_{\rho}
\bigr)^{\top}
\end{align}
is the vector of expectation values of the observables in $\mathcal{A}$. Hence,
\begin{align}
f(\mathcal{A})
=
\inf_{\boldsymbol{\alpha}\in\mathcal{C}(\mathcal{A})}
F(\boldsymbol{\alpha}) \; .
\end{align}

Quantum theory is intrinsically probabilistic and ultimately assigns probabilities to measurement outcomes. It is therefore natural to ask whether the expectation-value space generated by an arbitrary set of Hermitian observables can be represented, without loss of information, as a probability space generated by a suitable POVM. The following lemma answers this question in the affirmative: every MIBS is affinely isomorphic to the quantum probability space (QPS) induced by an explicitly constructed POVM.

\begin{lemma}[Affine isomorphism between MIBS and QPS]
\label{lem:affine_mapping}
Let $\mathcal{A}=\{A_i\}_{i=1}^{m}$ be a finite set of Hermitian operators, and let
\begin{align}
\mathcal{C}(\mathcal{A})
=
\left\{
\bigl(\alpha_1,\cdots,\alpha_m\bigr)^{\top}
\;\middle|\;
\alpha_i=\braket{A_i}_{\rho},\ 
\rho\in\mathcal{D}(\mathcal{H})
\right\}
\end{align}
be the associated MIBS. For each $i$, define the positive semidefinite operator
\begin{align}
B_i := A_i-\lambda_{\min}(A_i)\mathds{1} \geq 0 \; ,
\end{align}
and let
\begin{align}
\kappa
:=
\left\|
\sum_{i=1}^{m} B_i
\right\|_{\infty} \; .
\end{align}
Assume that $\kappa\neq0$. Then the operators
\begin{align}
E_i := \frac{1}{\kappa}B_i
=
\frac{1}{\kappa}
\bigl(
A_i-\lambda_{\min}(A_i)\mathds{1}
\bigr) \; ,
\qquad i=1,\cdots,m \; ,
\end{align}
together with
\begin{align}
E_0 := \mathds{1}-\sum_{i=1}^{m}E_i \; ,
\end{align}
form a POVM $\mathcal{E}=\{E_0,E_1,\cdots,E_m\}$. Let
\begin{align}
\mathcal{P}(\mathcal{E})
=
\left\{
\bigl(p_1,\cdots,p_m\bigr)^{\top}
\;\middle|\;
p_i=\braket{E_i}_{\rho},\ 
\rho\in\mathcal{D}(\mathcal{H})
\right\}
\subseteq \mathbb{R}^{m}
\end{align}
be the corresponding quantum probability space (QPS), where the coordinate associated with $E_0$ is omitted since it is fixed by normalization, $p_0=1-\sum_{i=1}^{m}p_i$. Then the map
\begin{align}
\Phi:\mathcal{C}(\mathcal{A})\longrightarrow \mathcal{P}(\mathcal{E}) \; ,
\qquad
\Phi_i(\boldsymbol{\alpha})
=
\frac{\alpha_i-\lambda_{\min}(A_i)}{\kappa},
\quad i=1,\cdots,m \; ,
\label{sm_eq:affine_map}
\end{align}
is an affine isomorphism, with inverse
\begin{align}
\Phi^{-1}_i(\mathbf{p})
=
\kappa\,p_i+\lambda_{\min}(A_i) \; ,
\quad i=1,\cdots,m \; .
\end{align}
Consequently, any optimization problem over $\mathcal{C}(\mathcal{A})$ can be equivalently transformed into an optimization problem over the QPS $\mathcal{P}(\mathcal{E})$.
\end{lemma}

\begin{proof}
\emph{Validity of the POVM.}
For each $i$, the operator $B_i=A_i-\lambda_{\min}(A_i)\mathds{1}$ is positive semidefinite by construction, so $E_i=B_i/\kappa\geq 0$ for $i=1,\cdots,m$. Since
\begin{align}
\sum_{i=1}^{m}E_i
=
\frac{1}{\kappa}\sum_{i=1}^{m}B_i \; ,
\end{align}
and $\kappa=\bigl\|\sum_i B_i\bigr\|_{\infty}$ equals the largest eigenvalue of the positive semidefinite operator $\sum_i B_i$, the operator $\sum_i E_i$ has largest eigenvalue equal to $1$. Hence
\begin{align}
0\leq \sum_{i=1}^{m}E_i \leq \mathds{1} \; ,
\end{align}
which guarantees that $E_0=\mathds{1}-\sum_{i=1}^{m}E_i\geq 0$. Together with $\sum_{i=0}^{m}E_i=\mathds{1}$, this shows that $\mathcal{E}=\{E_0,E_1,\cdots,E_m\}$ is a valid POVM.

\emph{Affine isomorphism.}
The map $\Phi$ in \cref{sm_eq:affine_map} is an invertible affine transformation of $\mathbb{R}^m$, consisting of a uniform scaling by $1/\kappa$ together with a translation by $-\lambda_{\min}(A_i)/\kappa$ in each coordinate. It therefore suffices to verify that $\Phi$ maps $\mathcal{C}(\mathcal{A})$ exactly onto $\mathcal{P}(\mathcal{E})$. For any state $\rho$ with $\boldsymbol{\alpha}=\boldsymbol{\alpha}_{\mathcal{A}}(\rho)$, the induced outcome probabilities satisfy
\begin{align}
p_i
=
\braket{E_i}_{\rho}
=
\frac{\braket{A_i}_{\rho}-\lambda_{\min}(A_i)}{\kappa}
=
\frac{\alpha_i-\lambda_{\min}(A_i)}{\kappa}
=
\Phi_i(\boldsymbol{\alpha}) \; .
\end{align}
Thus $\Phi\bigl(\mathcal{C}(\mathcal{A})\bigr)\subseteq\mathcal{P}(\mathcal{E})$, and since both spaces are generated by the \emph{same} set of states $\rho\in\mathcal{D}(\mathcal{H})$, equality holds. Conversely, the inverse map $\Phi^{-1}_i(\mathbf{p})=\kappa p_i+\lambda_{\min}(A_i)$ recovers $\alpha_i=\braket{A_i}_{\rho}$ from $p_i=\braket{E_i}_{\rho}$. As both $\Phi$ and $\Phi^{-1}$ are affine bijections, $\Phi$ is an affine isomorphism between $\mathcal{C}(\mathcal{A})$ and $\mathcal{P}(\mathcal{E})$.

\emph{Degenerate case.}
If $\kappa=0$, then $\sum_i B_i=0$, and since each $B_i\geq 0$ this forces $B_i=0$ for all $i$, i.e.\ every $A_i$ is proportional to the identity. In this case $\mathcal{C}(\mathcal{A})$ reduces to a single point and the optimization is trivial. Hence the construction above covers all nontrivial cases.
\end{proof}

This affine equivalence has an immediate consequence for the optimization problem. Define the transformed functional on the QPS by
\begin{align}
\widetilde{F}(\mathbf{p})
:=
F\bigl(
\kappa p_1+\lambda_{\min}(A_1),
\cdots,
\kappa p_m+\lambda_{\min}(A_m)
\bigr)
=
F\bigl(\Phi^{-1}(\mathbf{p})\bigr) \; .
\end{align}
Then the original problem can be rewritten as
\begin{align}
f(\mathcal{A})
=
\inf_{\boldsymbol{\alpha}\in\mathcal{C}(\mathcal{A})}
F(\boldsymbol{\alpha})
=
\inf_{\mathbf{p}\in\mathcal{P}(\mathcal{E})}
\widetilde{F}(\mathbf{p}) \; .
\label{sm_eq:mibs_qps_equivalence}
\end{align}
Therefore, optimizing a continuous functional of expectation values over the MIBS is equivalent to optimizing the corresponding transformed functional over a quantum probability space generated by a POVM.

This observation provides the conceptual bridge between the general observable-based formulation and the POVM-based probability-space formulation adopted. Although the optimization problems of interest are most naturally stated in terms of arbitrary Hermitian observables, \cref{lem:affine_mapping} shows that, up to an invertible affine transformation, no generality is lost by working with POVM-induced probability spaces.

%